\documentclass[a4paper,11pt]{article}
\usepackage[utf8]{inputenc}
\usepackage{amsmath,amssymb,amsfonts,amsthm,mathrsfs,graphicx}

\usepackage{multirow}
\usepackage{a4wide}
\usepackage{enumerate}
\usepackage{todonotes}

\usepackage{authblk}

\usepackage{xcolor}

\usepackage{hyperref}

\usepackage{a4wide}

\title{Positional $s$-of-$k$ games}

\author[1]{Eric Duch\^{e}ne}
\author[2]{Valentin Gledel}
\author[3]{Milo\v{s} Stojakovi\'c}

\affil[1]{Univ Lyon, CNRS, INSA Lyon, UCBL, Centrale Lyon, Univ Lyon 2, LIRIS, UMR5205, F-69622 Villeurbanne, France}
\affil[2]{Université Savoie Mont Blanc, CNRS UMR5127, LAMA, Chambéry, F-73000, France}
\affil[3]{Department of Mathematics and Informatics, Faculty of Sciences, University of Novi Sad, Serbia}
\date{}

\newtheorem{theorem}{Theorem}
\newtheorem{lemma}[theorem]{Lemma}

\newtheorem{corollary}[theorem]{Corollary}
\newtheorem{observation}[theorem]{Observation}

\newcommand{\scr}{\text{SC}}
\newcommand{\scrp}{\text{SC}_2}

\usepackage{tikz}
\newcommand{\rhombus}{\tikz \draw (0,0) -- (.23,0) -- (.38,.22) -- (0.15,.22) -- cycle;}
\newcommand{\hexagon}{\tikz \draw (0.1443,0) -- (0.07215,0.1249) -- (-0.07215,0.1249) -- (-0.1443,0) -- (-0.07215,-0.1249) -- (0.07215,-0.1249) -- cycle;}
\newcommand{\triangl}{\tikz \draw (0,0.25) -- (0.14435,0) -- (-0.14435,0) -- cycle;}

\newcommand{\hexa}{{\vcenter{\hbox{\;\!\scalebox{0.8}{\textnormal{\hexagon}}}}}}
\newcommand{\tria}{{\vcenter{\hbox{\scalebox{0.8}{\textnormal{\triangl}}}}}}
\newcommand{\squa}{{\vcenter{\hbox{\;\!\scalebox{0.75}{$\square$}}}}}
\newcommand{\rhom}{{\vcenter{\hbox{\scalebox{.8}{\rhombus}}}}}

\newcommand{\cF}{{\cal F}}
\newcommand{\cG}{{\cal G}}
\newcommand{\cP}{{\cal P}}

\begin{document}
\maketitle

\begin{abstract}

We introduce a general framework for positional games in which players score points by claiming a prescribed portion of each winning set, extending the notion of scoring Maker–Breaker games. In the scoring variant, Maker gains a point by fully claiming a winning set, while Breaker aims to minimize Maker’s total score. In this paper, we generalize these models for all $k$-uniform positional games by fixing an integer threshold $s\in\{1,2,\dots, k \}$ so that a player scores a point whenever she claims at least $s$ elements of a winning set of size $k$. We refer to this class as $s$-of-$k$ games. Such formulation allows for a flexible description of scoring objectives that appear in both theoretical models and real-life board games. 

We further investigate the impact of strategy restrictions on the achievable score. In particular, we analyze $s$-of-$k$ games both under optimal play, where the score is denoted by $\scr$, and under the additional constraint that Maker is restricted to a pairing strategy. The corresponding score in this setting is denoted by $\scrp$. While the unrestricted score captures the standard notion of optimal play in scoring positional games, the pairing-restricted score allows us to observe Maker's loss incurred by limiting her to these standard strategies.

We comprehensively study $s$-of-$k$ games played on regular grids, which provide a natural and uniform setting for illustrating the general framework. After developing several general tools for the analysis of both scores, we complement them by a number of ad-hoc strategies tailored for particular cases of these games, to obtain both upper and lower bounds for the two scores on triangular, square, rhombus and hexagonal grids.

%\mSt{Up to here another version of abstract.}
%\vspace*{5mm}
%
%
% This paper explores positional games, a well-known subclass of combinatorial games that includes popular recreational games like Tic-Tac-Toe and Hex. Recently, a scoring version of such games has been formalized, in which Maker's goal is to claim as many winning sets as possible. In the current paper, we generalize this family of games by introducing the set of $s$-of-$k$ games (with $1\leq s\leq k$), that are played on a $k$-uniform hypergraph, and where the objective of Maker consists in claiming at least $s$ vertices per winning set to get one point. For such games, we are interesting in the optimal score for each player (denoted by $\scr$), as well as the optimal score when using a pairing strategy (denoted by $\scrp$). Bounds on the optimal scores are first given in the general case, in particular by using a linear programming approach that provides upper bounds for $\scrp$. Then, we investigate the special case of regular structures corresponding to $k$-uniform hypergraphs such as triangular, square or hexagonal grids. Lower and upper bounds are given for both $\scr$ and $\scrp$, considering different kinds of strategies and techniques.
\end{abstract}

%\todo[inline]{the "classical" scoring variant -> it xas introduced not too long ago, maybe classical is not the right term} \mSt{True. Removed.}

{\color{red}
}

\section{Introduction}

\subsection{Scoring Positional Games}

Positional games are a well-studied subclass of combinatorial games, that includes popular recreational games like Tic-Tac-Toe and Hex. Structurally, a positional game is a hypergraph $(V,\mathcal{F})$, where $V$ is a finite set, the \emph{board} of the game, and $\mathcal{F} \subseteq 2^V \setminus \{\emptyset\}$ is a family of sets, the \emph{winning sets}. The game is played by two players who alternately claim the unclaimed elements of the board, until all of them are claimed.

When it comes to determining the winner, probably the most studied convention of positional games is the \emph{Maker-Breaker games}, already studied in the seminal papers of Erd\H{o}s and Selfridge~\cite{erdos-1973} and Chv\'atal and Erd\H{o}s~\cite{chvatal-1978}. The players are called Maker and Breaker, and there are two possible outcomes of the game; Maker's goal is to claim all elements of a winning set, while Breaker wins otherwise, i.e., if Breaker claims an element in every winning set. There are two books that cover the topic of positional games in more detail~\cite{beck-book, hkss}. 

The variant of Maker-Breaker games in which the aim of Maker is to fully claim \emph{as many winning sets as possible} appears in numerous places in the literature. It is discussed in several parts of the book~\cite{beck-book}, studied in~\cite{bedlu, dkms, rat, rat2, ss}, and it is more recently examined in more detail under the name \emph{scoring positional games} in~\cite{scoringPG}. %They can be played with the usual conventions of positional games. 
In such games, the game play is the same as in a standard positional game on a hypergraph $\mathcal{H}=\{V,\mathcal{F}\}$, but Maker gets a point each time she claims all the vertices of a hyperedge. The goal of Maker consists in maximizing her score, i.e., the number of hyperedges that she claims completely up until the end of the game. Breaker, on the other hand, aims at minimizing this score of Maker.

\emph{Majority games} are a further natural generalization of scoring positional games, that fits better with some real-life board games, like {\sc Conquete}\footnote{\url{https://escaleajeux.fr/jeu/conqd.0.0}} or {\sc Hedron}\footnote{\url{https://boardgamegeek.com/boardgame/11921/hedron}}. In such games, the game play remains the same, while Maker gets one point per hyperedge when she has claimed (strictly) more vertices than her opponent (who may be either Breaker, or another Maker player keeping track of his own score). Such games have been recently investigated by Bensmail and McInerney within the study of Vertex Capturing Game~\cite{vcg}. In their version, the players claim the edges of a graph and get a point each time they have the majority of incident edges of a vertex. They resolved the outcome of the game for paths, cycles, complete bipartite graphs and grids. 

Our current proposal is a way to define a general framework for the study of such games. Indeed, ``majority'' may not always be suitable to describe a desire of Maker to claim some predefined part of a winning set. Further more, when the size of the winning set is even, the case of ties is to be dealt with. Therefore, instead of simply considering the majority, setting a fixed number of elements, or a fraction of the size of the winning sets, turns out to be more suitable in full generality.

Several known games in the literature could be embedded within this framework. Given a fixed graph $H$, the $H$-game is a well-studied positional game played on the edge set of a base graph $E(G)$, where the winning sets are all copies of $H$ in $G$~\cite{beck-book, Hgame, hkss}. A standard approach to study the prospects for claiming a whole copy of a graph $H$ have been tied to the possibility to claim a (large enough) number of copies of $H^-$, a graph obtained by removing an edge from $H$, see, e.g.,~\cite{bedlu} or~\cite{ss}. This approach can be extended to looking back at $k$ moves before a possible Maker's win, counting copies of $H^{-k}$, a graph obtained by removing $k$ edges from $H$, for some integer $k$. Following the same spirit, the Maker-Breaker Hamiltonian Path game, played on edges, can be seen as an instance of the Maker-Breaker Hamiltonian cycle game where it suffices to get an $(n-1)$ part of a winning set to win (where $n$ is the number of vertices of the graph). In Kaplansky's $n$-in-a-line game, see~\cite{beck-kap, kr}, the winning sets are straight lines that typically, depending on the board, have much more than $n$ elements. Each of the players has the same goal, to be the first to claim a subset of certain cardinality, namely $n$, of a winning set, while the opponent has not claimed even a single element of that winning set.   \\

Therefore, in the current paper, we generalize the definition of scoring games given in \cite{scoringPG} with a constant parameter $s$ corresponding to the number of vertices required in a winning set to score one point. This particularly makes sense if all the winning sets have the same size, i.e., when the hypergraph is $k$-uniform. We refer to these games as \emph{$s$-of-$k$ games}. Note that a similar extension could be done with a \emph{ratio} of the size of the winning sets as the threshold for a score (this is essentially the same as $s$-of-$k$-games when the hypergraph of the game is uniform).

%In the current paper, we propose a generalization of their game, since in our version, the degree of a vertex is not bounded in the hypergraph (which is the case of Conquete or the vertex capturing game, where a vertex of the underlying hypergraph is in a most two hyperedges). But we do not embedd fully the vertex capturing game as we are k-uniform

\subsection{$s$-of-$k$ games}

Here we formally define $s$-of-$k$ games. Let $\cG$ be a positional game played on a $k$-uniform hypergraph $(V, \cF)$, and let $s\in \{1,2, \dots, k\}$ be an integer. Maker and Breaker alternately claim unclaimed elements of the board, with Maker going first, until all the elements are claimed. We are interested in the number of winning sets that have at least $s$ elements claimed by Maker -- Maker's goal is to maximize this number, while Breaker wants to minimize it. When both players play optimally, we refer to this number as the \emph{score} of the game $\cG$. At the end of the game, we say that a winning set is \emph{good} if the number of Maker's claimed elements in it is at least $s$, and we say that it is \emph{bad} if that number is below $s$.

The score of the game under optimal play of both players will be denoted by~$\scr (\cG, s)$. The game {\sc Incidence}, introduced in \cite{scoringPG}, that is played on the vertices of a graph, actually corresponds to the $2$-of-$2$ game, as Maker gets one point each time she claims both vertices of an edge. From this previous work, it is already known that computing $\scr (\cG, s)$ in the general case is a PSPACE-complete problem, since it is the case for {\sc Incidence}.

In addition, the following observation highlights that the roles of Maker and Breaker are somehow analogue.
\begin{observation} \label{obs:1}
$\scr (\cG, s) = |\cF| - \scr (\cG, k-s+1) + O(\Delta(\cF))$, where $\Delta(\cF)$ is the maximum degree of $\cF$.
\end{observation}

In other words, in an $s$-of-$k$ game, the objective of Breaker consists in ensuring that there are as many winning sets with $(k-s+1)$ vertices claimed by him. I.e., Breaker has the role of Maker in the $(k-s+1)$-of-$k$ game on the same hypergraph. Note that the $O(\Delta(\cF))$ term appears due to the fact that Maker starts both games, knowing that she in her first move cannot affect more than $\Delta(\cF)$ winning sets.\\

In the current study, we will provide an analysis of several $s$-of-$k$ games.
%according to the type of strategy that may be used. 
As pairing strategies are frequently used in the literature (due to their simplicity, as well as their versatility), we will devote special attention to the analysis of games when Maker is restricted to use a pairing strategy. It is interesting to study how far from the optimality pairing strategies can be. Such a discussion has already been considered in the book~\cite{beck-book} by Beck, who systematically studied games for which a player can not only win, but win by following a pairing strategy, which is clearly a stronger requirement. Recall that in a pairing strategy (say, for Maker), the vertices of the hypergraph are paired together in advance, in disjoint pairs. For each vertex $v$ claimed by Breaker during the game, the answer of Maker is the vertex that is paired with $v$. Note that pairing strategies correspond to non-adaptative strategies, since the responses of Maker are fixed before the game starts.  

In an $s$-of-$k$ game, when Maker is restricted to use a pairing strategy, the score of the game will be denoted by~$\scrp (\cG, s)$. This definition naturally yields the following observation.

\begin{observation} \label{obs:2}
$\scrp (\cG, s)\leq \scr (\cG, s)$.
\end{observation}

It is folklore (confirmed by numerous examples, see, e.g.,~\cite{beck-book}) that in Maker-Breaker games not all optimal strategies can be pairing strategies. However, proving that there does not exist an optimal pairing strategy in (scoring) positional game that we study here is not a straightforward task. In the context of $s$-of-$k$ games, we would like to show that $\scr (\cG, s)$ and $\scrp (\cG, s)$ have different values for some $\cG$ and $s$. By considering {\sc incidence} (i.e., the $2$-of-$2$ game) played on a cycle of size $14$, we have verified (with a relatively simple Python code that tests all the possible pairings, see \url{https://github.com/vgledel/test-pairing}) that $\scrp (C_{14}, 2)=2$, whereas it is known from~\cite{scoringPG} that $\scr (C_{14}, 2)=3$.

This result ensures that in general there is no equality between $\scr (\cG, s)$ and $\scrp (\cG, s)$. We conjecture that, for larger values of $n$, we still have a gap between $\scr (C_n, 2)$ and $\scrp (C_n, 2)$. 

\subsection{Exposition of the work}

In Section~\ref{sec:general} we present and develop several general results, that will be put to use later. The first one is a rewriting of the celebrated Erd\H{o}s-Selfridge Theorem~\cite{erdos-1973} in the context of scoring positional games. Two more results that we add are novel and specific to $s$-of-$k$ games, including a general approach that provides upper bounds for $\scrp (\cG, s)$ by analysing a random strategy for Breaker against a pairing Maker. This result will be useful to obtain, in many situations, non-trivial upper bounds on the best score available for Maker that is restricted to using pairing strategies. 

%\mSt{Fix numbers of sections, and the order!!}
In the rest of the paper, we will investigate particular $s$-of-$k$ games played on regular structures, that in particular imply the $k$-uniformity of the hypergraph. More precisely, we will explore $s$-of-$k$ games played on the vertices of several standard grids: the triangular grid where the winning sets are triangles $\cG_\tria$ (Section~\ref{sec:triangle}), the square grid $\cG_\squa$ where the winning sets are squares (Section~\ref{sec:square}), the rhombus grid where the winning sets are rhombi on the triangular grid $\cG_\rhom$ (Section~\ref{sec:rhombus}), and finally the hexagonal grid $\cG_\hexa$ where the winning sets are all hexagons (Section~\ref{sec:hexagon}). \\

Our study will deal with finite but large grids, where the number of winning sets is denoted by $n$. We generally note that most of the strategies we develop will be tailored to work at a distance from the boundary of the grid.
We will also routinely not discuss the particular shape of the grid, as this will not affect our analysis. Our only assumption throughout the paper is that all the grids are ``two dimensional'', i.e., that the number of winning sets that lie at a constant distance to the boundary is asymptotically small compared to $n$. In the expression of our results, this recurrent $o(n)$ error in the score will be mentioned but we will not give it any further attention. 

All the results provided in sections 3 through 6 can be found in the summarizing Table~\ref{tab:results}. It contains lower and upper bounds for two types of strategies -- optimal and pairing, for all four types of grids, and all relevant values of $s$. %Throughout the paper, we will explore various techniques to prove these results, ranging from original tilings to counting arguments. 

\begingroup
\setlength{\tabcolsep}{4pt}

\begin{table}[h]
    \centering
\begin{tabular}{|c|c|c|c|c|c|}\hline
  \multicolumn{2}{|c|}{\multirow{2}{*}{Problem}} & \multicolumn{2}{|c|}{$\scr$} & \multicolumn{2}{|c|}{$\scrp$} \\ \cline{3-6}
  \multicolumn{2}{|c|}{} & lower bound &  upper bound & lower bound &  upper bound \\ \hline
  \multirow{3}{*}{$\cG_\tria$} & $s=1$ & $\geq 7n/8$ (E-S) & $\leq 27n/28$ (Th~\ref{thm:Maker-s3}) & $\geq 3n/4$ (Th~\ref{thm:scrp-trig-s1}) & $\leq 15n/16$ (Th~\ref{thm:upper_pairing}) \\ \cline{2-6}
   & $s=2$ & $n/2$ (Lm~\ref{lem:kodd}) & $n/2$ (Lm~\ref{lem:kodd})& $\geq 3n/8$ (Th~\ref{thm:scrp-trig-s2}) & $\leq n/2$ (Th~\ref{thm:upper_pairing})  \\ \cline{2-6}
   & $s=3$ & $\geq n/28$ (Th~\ref{thm:Maker-s3}) & $\leq n/8$ (E-S) & $\geq n/28$ (Th~\ref{thm:Maker-s3}) & $\leq n/8$ (Th~\ref{thm:upper_pairing}) \\ \hline
  \multirow{4}{*}{$\cG_\squa$} & $s=1$ & \multicolumn{4}{|c|}{$\scr(\cG_\squa, 1) = \scrp(\cG_\squa, 1) = n$ (Obs~\ref{obs:square-s14})} \\ \cline{2-6}
   & $s=2$ & $\geq 2n/3$ (Th~\ref{thm:square-s2}) & $\leq 13n/15$ (Th~\ref{thm:SCsquares=3}) & $\geq 2n/3$ (Th~\ref{thm:square-s2})& $\leq 3n/4$ (Th~\ref{thm:upper_pairing}) \\ \cline{2-6}
   & $s=3$ & $\geq 2n/15$ (Th~\ref{thm:SCsquares=3}) & $\leq n/3$ (Th~\ref{thm:square-s2}) & $\geq n/8$ (Th~\ref{thm:fractal-pairing}) & $\leq 5n/16$ (Th~\ref{thm:upper_pairing}) \\ \cline{2-6}
   & $s=4$ & \multicolumn{4}{|c|}{$\scr(\cG_\squa, 4) = \scrp(\cG_\squa, 4) = 0$ (Obs~\ref{obs:square-s14})} \\ \hline
  \multirow{4}{*}{$\cG_\rhom$} & $s=1$ & $\geq 15n/16$ (E-S) & $\leq 77n/78$ (Th~\ref{thm:rhombus-maker-s4}) & $\geq 11n/12$ (Th~\ref{thm:rhombus-pairing-s1}) & $\leq 23n/24$ (Th~\ref{thm:upper_pairing}) \\ \cline{2-6}
   & $s=2$ & $\geq 19n/36$ (Th~\ref{thm:rhom-pairing-LB}) & $\leq 89n/96$ (Th~\ref{thm:rhombus-pairing-s3}) & $\geq 19n/36$ (Th~\ref{thm:rhom-pairing-LB}) & $\leq 71n/96$ (Th~\ref{thm:upper_pairing}) \\ \cline{2-6}
   & $s=3$ & $\geq 7n/96$ (Th~\ref{thm:rhombus-pairing-s3}) & $\leq 17n/36$ (Th~\ref{thm:rhom-pairing-LB}) & $\geq 7n/96$ (Th~\ref{thm:rhombus-pairing-s3}) & $\leq 5n/16$ (Th~\ref{thm:upper_pairing}) \\ \cline{2-6}
   & $s=4$ & $\geq n/78$ (Th~\ref{thm:rhombus-maker-s4}) & $\leq n/16$ (E-S) & $\geq 0$ & $\leq n/16$ (Th~\ref{thm:upper_pairing}) \\ \hline
  \multirow{6}{*}{$\cG_\hexa$} & $s=1$ & \multicolumn{4}{|c|}{$\scr(\cG_\hexa, 1) = \scrp(\cG_\hexa, 1) = n$ (Obs~\ref{obs:hexa-s1256})} \\ \cline{2-6}
   & $s=2$ & \multicolumn{4}{|c|}{$\scr(\cG_\hexa, 2) = \scrp(\cG_\hexa, 2) = n$ (Obs~\ref{obs:hexa-s1256})} \\ \cline{2-6}
   & $s=3$ & $\geq 3n/4$ (Th~\ref{thm:hex-s3-pairing})& $\leq 5n/6$ (Th~\ref{thm:hex-s4-Maker-nopairing}) & $\geq 3n/4$ (Th~\ref{thm:hex-s3-pairing})& $\leq 5n/6$ (Th~\ref{thm:hex-s4-Maker-nopairing})\\ \cline{2-6}
   & $s=4$ & $\geq n/6$ (Th~\ref{thm:hex-s4-Maker-nopairing}) & $\leq n/4$ (Th~\ref{thm:hex-s3-pairing}) & $\geq n/10$ (Th~\ref{thm:hex-s4-Maker-pairing}) & $\leq n/4$ (Th~\ref{thm:hex-s3-pairing}) \\ \cline{2-6}
   & $s=5$ & \multicolumn{4}{|c|}{$\scr(\cG_\hexa, 5) = \scrp(\cG_\hexa, 5) = 0$ (Obs~\ref{obs:hexa-s1256})} \\ \cline{2-6}
   & $s=6$ & \multicolumn{4}{|c|}{$\scr(\cG_\hexa, 6) = \scrp(\cG_\hexa, 6) = 0$ (Obs~\ref{obs:hexa-s1256})} \\ \hline
\end{tabular}
    \caption{The different bounds on $\scr$ and $\scrp$ found in this paper (omitting the $o(n)$ term for conciseness). In the table, E-S refers to Theorem~\ref{thm:ES}, Erd\H os-Selfrige Theorem applied to scoring positional games.}
    \label{tab:results}
\end{table}

\endgroup

% moved!
%
% \begin{prop}
% Let $G$ be a rectangular grid with $n$ squares. We have $n/8\leq Ms(G)$.
% \end{prop}
% \begin{proof}
% Todo and todraw. Proof with the fractal tiling.
% \end{proof}

% \begin{prop}
% Let $G$ be a rectangular grid with $n$ squares. We have $2n/15\leq Ms(G)$.
% \end{prop}
% \begin{proof}
% Todo. We first prove that $Ms(G_{3,5})=Bs(G_{3,5})=2$. Then we tile $G$ with $G_{3,5}$ grids and Maker applies her strategy on each $G_{3,5}$. The squares between the tiles are lost for Maker. Thus she scores 2 points for every set of 15 squares.
% \end{proof}

% \todo[inline]{Lower bound for $SC$, s=2}

% \section{ToDo}

% \begin{itemize}
%     \item $k\times n$ . With $k=2, Ms(G)=0$
%     \item Erdos Selfridge like bounds
%     \item Triangular grids
% \end{itemize}

%%%%%%%%%%%%%%%%%%%%%%%%%%%%%%%%%%%%%%%%%%%%%%%%
\section{General results} \label{sec:general}

The first theorem is a rewriting of Erd\H{o}s-Selfridge Theorem, as done in~\cite[Theorem 10]{scoringPG}, adjusted to our notation. It yields bounds on the optimal score when $s\in \{1,k\}$. 

\begin{theorem}[Erd\H os-Selfridge Theorem,~\cite{scoringPG}]
\label{thm:ES}
Let $\mathcal \cG=(V,\mathcal F)$ be a $k$-uniform hypergraph with $n$ hyperedges, we have $\scr(\cG,k) \le \underset{e \in \mathcal{F}}{\sum} 2^{-|e|} = n2^{-k}$ and $\scr(\cG,1) \ge  n(1-2^{-k})$.
\end{theorem}

Next, in Theorem~\ref{thm:upper_pairing}, we utilize a probabilistic analysis to prove the existence of a strategy for Breaker against a pairing strategy in any game that is uniform and regular. As it turns out, for several of the games that we consider, the score obtained in this way beats all currently known deterministic strategies. The applications of this theorem can be found in the last column of Table~\ref{tab:results}.

%\mSt{Should we promote Lemma~\ref{lem:upper_pairing} into a Theorem?}

\begin{theorem}\label{thm:upper_pairing}
Let $H$ be a hypergraph with $n$ winning sets such that each winning set contains $k$ vertices, each vertex belongs to $\ell$ winning sets and each pair of vertices belongs to at most $q$ winning sets. Let $s \in  [k]$. For all $i \in \{k-2\lfloor\frac{k}{2}\rfloor,\dots,k-2, k\}$, let $X_i$ be a random variable distributed according to the binomial distribution $\mathcal B (i,1/2)$, and let %$p_{s,i}$ be the probability 
$p_{s,i} := \mathbb P (\frac{k-i}{2}+X_i \geq s)$. 

Then $\scrp(H, s) \leq z\cdot n$, where $z$ is the optimal value of the following linear program
\begin{equation} \label{eq:LP}
\begin{split}
\text{maximize }  & \sum_i p_{s,i} \cdot m_i, \\
\text{subject to } & \sum_i m_i = 1,\\
                 & \sum_i i\cdot m_i \geq \frac{k (\ell-q)}{\ell},
\end{split}
\end{equation}
with variables $m_i\geq 0$, $i \in \{k-2\lfloor\frac{k}{2}\rfloor,\dots,k-2, k\}$.
\end{theorem}

\begin{proof}
Since each winning set has $k$ vertices and each vertex belongs to $\ell$ winning sets, there are $\frac{k\cdot n}{\ell}$ vertices.

Let $\cP$ be a pairing of $H$. We denote by $M_i$ the set of winning sets that contain exactly $i$ vertices not paired within the winning set, where $i \in \{k-2\lfloor\frac{k}{2}\rfloor,\dots,k-2, k\}$. Let $m_i = \frac{|M_i|}{n}$. We have $\sum_i m_i = 1$. Moreover, for each pair $(u,v)$ of $\cP$, there are at most $q$ winning sets containing both $u$ and $v$. Therefore, for each vertex $u$, it is in at least $\ell -q$ winning sets not containing its paired vertex. This further implies $\sum_i |M_i| \geq (\ell -q)\cdot |V| = (\ell -q)\frac{k\cdot n}{\ell}$, and $\sum_i m_i \geq (\ell -q)\frac{k}{\ell}$.

Playing the $s$-out-of-$k$ game, if Maker plays according to the pairing $\cP$, then Breaker decides which of the two vertices in each pair Maker claims. Suppose that Breaker decides that by flipping a fair coin, independently for each pair. Then, for each $i$ and for each winning set $e\in M_i$, the number of vertices played by Maker in the winning set is $\frac{k-i}{2}+X_i$, where $X_i$ is a random variable following the binomial distribution $\mathcal B (i,1/2)$. Therefore, the probability that Maker scores the winning set $e$ is $p_{s,i}$, and the expected total number of winning sets won by Maker at the end of the game is $\sum_i (n\cdot m_i\cdot p_{s,i}) = n \sum_i (m_i\cdot p_{s,i})$.  

Since this probabilistic strategy of Breaker has the above mentioned expected Maker's score, there has to exist a deterministic strategy for Breaker allowing at most the same score. Hence, this number is an upper bound for the score of Maker for the given values of $m_i$, and Maker cannot score a better fraction of the total number of winning sets than the optimal value of the linear program~(\ref{eq:LP}).
\end{proof}

This lemma allows us to obtain the results presented in Table~\ref{tab:upper_bounds}, yielding upper bounds for the values of $\scrp$.

\begin{table}[ht]
    \centering
{\renewcommand{\arraystretch}{1.2}
\begin{tabular}{|c|c|l|c|}\hline     
    \multicolumn{2}{|c|}{Problem}& Function to be maximized & Bound  \\ \hline
    \multirow{3}{*}{Triangular grid} & $s =1$ & $m_1 + \frac{7}{8}m_3$ & $\scrp(\cG_\tria, 1) \leq \frac{15}{16}n$\\ \cline{2-4}
    & $s =2$ & $\frac{1}{2}m_1 + \frac{1}{2}m_3$ & $\scrp(\cG_\tria, 2) \leq \frac{1}{2}n$\\ \cline{2-4}
    & $s =3$ & $\frac{1}{8}m_3$ & $\scrp(\cG_\tria, 3) \leq \frac{1}{8}n$\\ \hline
    \multirow{4}{*}{Square grid} & $s =1$ & $m_0 + m_2 + \frac{15}{16}m_4$ & $\scrp(\cG_\squa, 1) \leq n$\\ \cline{2-4}
    & $s =2$ & $m_0+ \frac{3}{4}m_2 + \frac{1}{2}m_4$ & $\scrp(\cG_\squa, 2) \leq \frac{3}{4}n$\\ \cline{2-4}
    & $s =3$ & $\frac{1}{4}m_2 + \frac{5}{16}m_4$ & $\scrp(\cG_\squa, 3) \leq \frac{5}{16}n$\\ \cline{2-4}
    & $s =4$ & $\frac{1}{16}m_4$ & $\scrp(\cG_\squa, 4) \leq \frac{1}{16}n$\\ \hline
    \multirow{4}{*}{Rhombus grid} & $s =1$ & $m_0 + m_2 + \frac{15}{16}m_4$ & $\scrp(\cG_\rhom, 1) \leq \frac{23}{24}n$\\ \cline{2-4}
    & $s =2$ & $m_0+ \frac{3}{4}m_2 + \frac{1}{2}m_4$ & $\scrp(\cG_\rhom, 2) \leq \frac{71}{96}n$\\ \cline{2-4}
    & $s =3$ & $\frac{1}{4}m_2 + \frac{5}{16}m_4$ & $\scrp(\cG_\rhom, 3) \leq \frac{5}{16}n$\\ \cline{2-4}
    & $s =4$ & $\frac{1}{16}m_4$ & $\scrp(\cG_\rhom, 4) \leq \frac{1}{16}n$\\ \hline
    \multirow{6}{*}{Hexagonal grid} & $s =1$ & $m_0 + m_2 + m_4 + \frac{63}{64}m_6$ & $\scrp(\cG_\hexa, 1) \leq n$\\ \cline{2-4}
    & $s =2$ & $m_0 + m_2 + \frac{15}{16}m_4 + \frac{57}{64}m_6$ & $\scrp(\cG_\hexa, 2) \leq n$\\ \cline{2-4}
    & $s =3$ & $m_0 + \frac{3}{4}m_2 + \frac{11}{16}m_4 + \frac{21}{32}m_6$ & $\scrp(\cG_\hexa, 3) \leq \frac{85}{96}n$\\ \cline{2-4}
    & $s =4$ & $\frac{1}{4}m_2 + \frac{5}{16}m_4 + \frac{11}{32}m_6$ & $\scrp(\cG_\hexa, 4) \leq \frac{11}{32}n$\\ \cline{2-4}
    & $s =5$ & $\frac{1}{16}m_4 + \frac{7}{64}m_6$ & $\scrp(\cG_\hexa, 5) \leq \frac{7}{64}n$\\ \cline{2-4}
    & $s =6$ & $\frac{1}{64}m_6$ & $\scrp(\cG_\hexa, 6) \leq \frac{1}{64}n$\\ \hline
\end{tabular}
}
    \caption{Upper bounds obtained with Theorem~\ref{thm:upper_pairing}.}
    \label{tab:upper_bounds}
\end{table}

%\todo[inline]{Valentin: "rhombus grid"? Also there's maybe something to say about the symmetry between the rhomus and the square cases}

%These results are upper bounds and as we will see in the following sections, the values of $\scrp(\cG_\squa, 4)$, $\scrp(\cG_\hexa, 5)$ and $\scrp(\cG_\hexa, 6)$ are actually equal to $0$. 

To wrap up this section, the following lemma considers the particular case of $k$ odd and $s$ equal to half this value. Under such conditions, the roles of Maker and Breaker are essentially the same, since each of them aims at claiming exactly half the size of a winning set in as many winning sets as possible. This allows for a double application of a standard strategy stealing argument, to obtain tight bounds on $\scr$.

\begin{lemma}\label{lem:kodd}
Let $k$ be an odd integer and $s=(k+1)/2$. Let $n$ be the number of winning sets in a $k$-uniform hypergraph $\cG$ with $o(n)$ maximum degree. Then we have  $\scr(\cG,s)=n/2+o(n)$.
\end{lemma}

\begin{proof}
From Observation~\ref{obs:1}, the roles of both Maker and Breaker are to score as many winning sets having at least $s$ vertices. 
If we assume, for a contradiction, that Breaker can achieve more than $n/2$ winning sets with at least $s$ vertices in them, then Maker can steal this strategy and employ it as first, ignoring her first move. That way she scores more than $n/2$ winning sets, a contradiction. 

On the other hand, Breaker, by playing second, can imagine he is the first player and steal the strategy of Maker. For his first move, he imagines that he plays it, even though he does not. With such a strategy, Breaker loses at most the winning sets containing that imagined first move, i.e., an $o(n)$ value. 

Hence, asymptotically, each player can score $n/2+o(n)$ winning sets.
\end{proof}

\section{Triangular grid, triangle, $k=3$}\label{sec:triangle}

\subsection{Triangle with $s=1$}

%$\blacktriangleright$ State of the art: $3n/4 \leq \scrp(\cG_\tria, 1) \leq 15n/16 $, $ 7n/8 \leq \scr(\cG_\tria, 1) \leq 27n/28$.

From Theorem~\ref{thm:upper_pairing} we get $\scrp(\cG_\tria, 1) \leq 15n/16 + o(n)$, from Theorem~\ref{thm:Maker-s3} we have $\scr(\cG_\tria, 1) \leq 27n/28 + o(n)$, and from Erd\H{o}s-Selfridge Theorem we get $\scr(\cG_\tria, 1) \ge 7n/8 + o(n)$.

\begin{theorem}\label{thm:scrp-trig-s1}
$3n/4 \leq \scrp(\cG_\tria, 1)$.
\end{theorem}

\begin{proof}
Maker pairs the vertices along the edges of the grid marked green in Figure~\ref{fig:matching}. Clearly, she will touch all the triangles marked red, as each of them has paired vertices on one of its sides.
\begin{figure}[htb]
    \centering
    \includegraphics[scale=2]{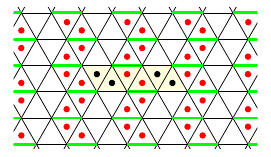}
    \caption{Pairing for Maker, triangle game with $s=1$.}
    \label{fig:matching}
\end{figure}

Further more, in the block marked yellow there are four more triangles, marked black. No matter which two vertices of the two (dark) green pairs within the block Maker claims, at least two out of the four black triangles will be touched. As all the remaining triangles can be partitioned into such groups of four (like four triangles marked black), Maker will touch at least half of the remaining triangles.

Altogether, Maker will touch at least $n/2 + 1/2\cdot n/2 = 3n/4$ triangles.
\end{proof}

\subsection{Triangle with $s=2$}

%\todo[inline]{Change it by using lemma 5}

%$\blacktriangleright$ State of the art:  $3n/8 \leq \scrp(\cG_\tria, 2) \leq n/2$,  $\scr(\cG_\tria, 2) \leq 5n/8$.

From Lemma~\ref{lem:kodd}, we get $\scr(\cG_\tria, 2) = n/2 + o(n)$.

\begin{theorem}\label{thm:scrp-trig-s2}
$3n/8\leq \scrp(\cG_\tria, 2)$.
\end{theorem}
\begin{proof}
\begin{figure}[ht]
    \centering
    \begin{minipage}[c]{0.5\linewidth}
        \includegraphics[width = 8cm]{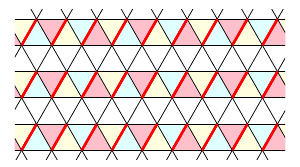}
    \end{minipage}
    \qquad
    \begin{minipage}[c]{0.4\linewidth}
        \includegraphics[width = 6cm]{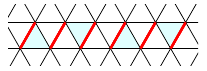}
    \end{minipage}
    \caption{The pairing for the triangular grid and $s=2$, and two pairs of matched triangles of the same color.}
    \label{fig:triangle-pairing-s2}
\end{figure}

Maker pairs the vertices as depicted in red in Figure~\ref{fig:triangle-pairing-s2} on the left. Among the triangles with one side paired that are colored with the same color, either yellow, red or blue, we are guaranteed to score at least half of them. Indeed, there is a matching among the triangles of each respective color, such that two matched triangles have two of their vertices paired, while the third one is paired with a vertex from the matched triangle, see Figure~\ref{fig:triangle-pairing-s2} on the right. 

On the remaining triangles (the uncolored triangles, lying on the strips between the paired edges), knowing that the total average of Maker's vertices per triangle is $1.5$, we can establish a lower bound for the score assuming the worst-case distribution -- all scored triangles have 3 Maker's vertices, and the rest have 1 Maker's vertex. 
Hence, the total score must be at least $n/2\cdot 1/2 + n/2\cdot 1/4 = 3n/8$.
\end{proof}

\subsection{Triangle with $s=3$}

%$\blacktriangleright$ State of the art: $n/28 \leq \scrp(\cG_\tria, 3) \leq \scr(\cG_\tria, 3) \leq n/8$. \\

From Theorem~\ref{thm:upper_pairing} we get $\scrp(\cG_\tria, 3) \leq n/8 + o(n)$ and from Erd\H{o}s-Selfridge Theorem, we get $\scr(\cG_\tria, 3) \leq n/8 + o(n)$.

% \begin{lemma} \label{lem:matching}
% $\scr(\cG_\tria, 3) \leq n/4$.
% \end{lemma}

% \begin{theorem} \label{thm:Maker-s3}
% $n/48 \leq \scrp(\cG_\tria, 3)$
% \end{theorem}
% \begin{proof}
% Maker applies the pairing in Figure~\ref{fig:triangle-pairing-s3}, where the pairs are denoted by both red and green lines. In the horizontal sequence of grid vertices that are paired by green lines there must be two consecutive Maker's vertices in every eight consecutive segments. Red pairing then ensures that one of the two triangles on that pair of Maker's vertices will be completely Maker's. 

% That esures one Maker's triangle in every 48 triangles.
% \begin{figure}[htb]
%     \centering
%     \includegraphics[scale=0.7]{figures/triangle-s3-pairing-lb.png}
%      \caption{Pairing for Maker, triangle game with $s=3$}
%     \label{fig:triangle-pairing-s3}
% \end{figure}
% \end{proof}

\begin{theorem} \label{thm:Maker-s3}
$n/28\leq \scrp(\cG_\tria, 3)$.
\end{theorem}
\begin{proof}
Maker applies the pairing in Figure~\ref{fig:triangle-pairing-s3} where the hexagon centers, marked red, are arbitrarily paired among themselves, and other than that each hexagon has the same pairing on the remaining vertices.
\begin{figure}[htb]
    \centering
    \includegraphics[scale=1.7]{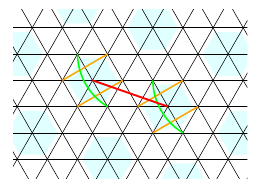}
     \caption{Pairing for Maker, triangle game with $s=3$.}
    \label{fig:triangle-pairing-s3}
\end{figure}

First, note that out of every two paired red vertices Maker will claim one. Let us concentrate on the hexagon around that vertex. Out of the two paired green vertices Maker will claim one, resulting in two neighboring Maker's vertices, one red and one green. Finally, out of two paired orange vertices that can complete a Maker's triangle Maker will claim one, thus claiming a triangle.

Hence, we have one Maker's triangle per two hexagons.
\end{proof}

%%%%%%%%%%%%%%%%%%%%%%%%%%%%%%%%%%%%%%%%%%%%%
\section{Square grid, square, $k=4$}\label{sec:square}

\subsection{Square with $s\in\{1,4\}$}

\begin{observation}
\label{obs:square-s14}
$\scr(\cG_\squa, 1) = \scrp(\cG_\squa, 1) = n$ and $\scr(\cG_\squa, 4) = \scrp(\cG_\squa, 4) = 0$.
\end{observation}

\begin{proof}
Using Observation~\ref{obs:1}, it is sufficient to provide a pairing for $s=1$ such that every square of the grid contains a pair. The pairing is the following: we color the grid squares black/white in a checkerboard fashion; every white square has paired vertices on its left edge, and every black square has paired vertices on its right edge.
\end{proof}

\subsection{Square with $s=2$}

%$\blacktriangleright$ State of the art: $2n/3 \leq \scrp(\cG_\squa, 2) \leq 3n/4$, $\scr(\cG_\squa, 2)\leq 13n/15$.

From Theorem~\ref{thm:upper_pairing} we get $\scrp(\cG_\squa, 2)\leq 3n/4 + o(n)$ and from Theorem~\ref{thm:SCsquares=3} we get $\scr(\cG_\squa, 2)\leq 13n/15 + o(n)$.

\begin{theorem} \label{thm:square-s2}
$\scrp(\cG_\squa, 2) \geq 2n/3$.
\end{theorem}
\begin{proof}
Figure~\ref{fig:squares2} shows a pairing, in green, for Maker that guarantees to get at least $2/3$ of the total number of squares. 
\begin{figure}[htb]
    \centering
    \includegraphics[scale=1.4]{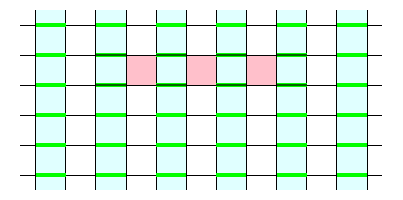}
    \caption{Pairing for Maker, square game with $s=2$.}
    \label{fig:squares2}
\end{figure}
Straightforwardly, this pairing guarantees that Maker will claim each blue square, i.e.,~half of the total number of squares. In addition, we will prove that for the triple of squares marked in red, she will score at least one of them. If we look at the eight green pairs, marked with black lines in the figure, touching the red squares, they cover the total of 16 vertices, out of which 12 are on the red squares. As the 16 vertices are paired Maker will claim half of them, eight vertices. Four of them have to be on the red squares and by pigeonhole principle there are two Maker's vertices on one of the red squares, which is enough for a score.

As all the grid squares that are not blue can be partitioned into triples congruent to the red triple, Maker will score at least a third of them. This ensures that Maker is guaranteed to score the total of $4n/6$ squares.
\end{proof}

\subsection{Square with $s=3$}
%$\blacktriangleright$ State of the art: $n/8\leq \scrp(\cG_\squa, 3)\leq 5n/16$, $2n/15\leq \scr(\cG_\squa, 3) \leq n/3$. \\

From Theorem~\ref{thm:upper_pairing} we get $\scrp(\cG_\squa, 3)\leq 5n/16 + o(n)$, and from Theorem~\ref{thm:square-s2} we get $\scr(\cG_\squa, 3)\leq n/3+ o(n)$.

\begin{theorem} \label{thm:fractal-pairing}
$\scrp(\cG_\squa, 3) \geq n/8 + o(n)$.
\end{theorem}

\begin{proof}
We will give a pairing for Maker. The set of pairs will be partitioned to levels.

For the first level we subdivide the initial square grid into the 3-by-3 squares, as depicted in Figure~\ref{fig:squares-s3-pairing} (left), and pair the vertices as marked by the blue segments. Note that the vertices marked red are not paired on this level, they are yet to be paired on higher levels.

Now on the second level we look at the grid consisting of the red vertices, which are conveniently arranged in a square grid. We again subdivide that square grid into the 3-by-3 squares, and apply the same pairing recursively. Again some vertices are not paired, we take them to the one higher level, repeating the process.
\begin{figure}[ht]
    \centering 
    \begin{minipage}[c]{0.55\linewidth}
        \includegraphics[scale=1.4]{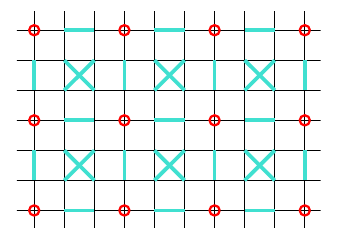}
    \end{minipage}
    $\qquad$
    \begin{minipage}[c]{0.3\linewidth}
        \includegraphics[scale=1.8]{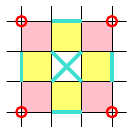}
    \end{minipage} 
    \caption{Pairing on the square grid with $s=3$. Left: The pairing on one level, where the vertices marked red are not paired. Right: The \emph{side} squares (yellow) and the \emph{corner} squares (red). }
    \label{fig:squares-s3-pairing}
    \end{figure}

Now looking at the first level, in each 3-by-3 square we observe the four \emph{side} squares, marked yellow in Figure~\ref{fig:squares-s3-pairing} (right). Note that one of them will have at least three Maker's elements and will increase the score by one.

Furthermore, note that the four \emph{corner} squares, marked red in Figure~\ref{fig:squares-s3-pairing} (right), form a partition of the 16 vertices in a 3-by-3 square. Hence, if at least three out of four vertices marked red are claimed by Maker, she will have more than half of all vertices in this 3-by-3 square. To see this note that the non-red vertices are paired within the 3-by-3 square. Then, clearly, at least one of the four red squares must have more than half vertices claimed by Maker, and that also increases Maker's score by at least one.

Now for the final argument, we look at the levels one by one starting from the last level. Each 3-by-3 square on a fixed level has at least one yellow square with Maker's advantage on the corners. Further more, if there is a Maker's advantage on the red corners of that 3-by-3 square, then there is also one red square with Maker's advantage. To see this, note that the 16 vertices are partitioned by the four red squares, and all of them are paired except the four red ones. This way, the Maker's advantage is propagated down the level hierarchy.

Summing the score up on all levels, each 3-by-3 square (on any level) contributes to the total Maker's score by at least 1. The number of 3-by-3 squares on the $i$-th level is $n/9^i$, so Maker's score must be at least
$$ \sum_{i=1}^{\lfloor\log_9(n)\rfloor} \frac{n}{9^i} = \frac{n}{8} + o(n).$$
\end{proof}

For this case, we can get a better bound for Maker that does not use a pairing strategy. We consider the finite rectangular grid $G_{n,m}$ with $n$ rows and $m$ vertices. The next lemma gives a lower bound on the score for $3\times 5$ grids.

\begin{lemma}
$\scr(G_{3,5},3)\geq 2$. This bound is valid no matter who starts. 
\end{lemma}

\begin{proof}
Let $x_i,y_i,z_i$, $i=1,\dots, 5$, be the vertices of $G_{3,5}$ arranged as in Figure~\ref{fig:label3_5}. \\

\begin{figure}[ht]
    \centering
    \includegraphics[scale=1.15]{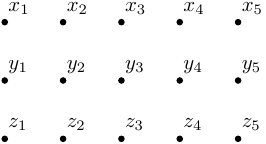}
    \caption{A subgrid $G_{3,5}$.}
    \label{fig:label3_5}
\end{figure}
Since scoring positional games belong to Milnor's universe~\cite{scoringPG}, they are non-zugzwang, which means that it is always better for a player to move rather than pass. Consequently, it suffices to prove the result when Breaker plays first. 

As a preliminary remark, note that on a $G_{3,2}$ grid, if Maker claims the two central vertices (before any move of Breaker), then she scores two points in the end by pairing the remaining vertices in the obvious way.

Assume first that Breaker does not play the first move in the central column. Without loss of generality, he plays a vertex $x_i$ or $y_i$ with $i=1,2$. Then Maker answers by claiming $y_4$. If Breaker now aims at avoiding Maker to get two points, he must prevent her to get the two central vertices of each horizontal and vertical subgrid $G_{3,2}$ including $y_4$. This is not possible because the intersection of all such patterns in the remaining board is empty. 

Assume now that Breaker plays $x_3$ (or $z_3$ equivalently). Then Maker claims $y_2$. As depicted by Figure~\ref{fig:G2_3move1}, there are now two threats for Breaker: indeed, if Maker plays $y_1$ or $z_2$, she would manage to have a $G_{3,2}$ grid with the two central vertices claimed. Hence Breaker must play in the intersection of the these two colored subgrids, meaning that she plays either $y_1$, $z_1$ or $z_2$. By symmetry, Maker now plays $y_4$, and Breaker is forced to answer on either $y_5$, $z_4$ or $z_5$. Maker now plays $y_3$ and Breaker is forced to answer $z_3$ (otherwise Maker plays it and wins two points). Finally Maker claims $x_2$ and gets one point, and for her next move, she claims either $x_1$ or $x_4$ to get the second point.

\begin{figure}[ht]
    \centering
    \includegraphics[scale=1.15]{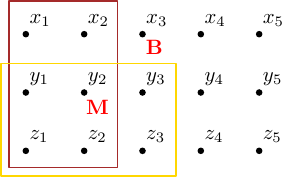}
    \caption{Solving the $G_{3,5}$ game, case $s=3$, first move $x_3$.}
    \label{fig:G2_3move1}
\end{figure}

The last case is when Breaker first plays $y_3$. Maker answers $y_2$, as shown by Figure~\ref{fig:G2_3move2}. There is now a threat on the left subgrid $G_{3,2}$ where Breaker is forced to play. Let $b_1$ be the vertex played by Breaker on this subgrid. Now Maker claims $y_4$, and by symmetry, let $b_2$ be the vertex played by Breaker on the right subgrid. 

\begin{figure}[ht]
    \centering
    \includegraphics[scale=1.15]{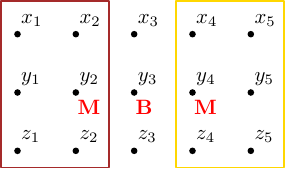}
    \caption{Solving the $G_{3,5}$ game, case $s=3$, first move $y_3$.}
    \label{fig:G2_3move2}
\end{figure}

If $b_1\neq y_1$, then Maker now plays $y_1$ and pairs two vertices of a square inside this pattern to ensure one point. Breaker is now forced to claim the last vertex of the subgrid, otherwise he loses a second point. Now Maker focuses on the right subgrid. If $b_2\neq y_5$, then Maker plays $y_5$ and will get one point with the similar strategy. Otherwise, Maker plays $x_4$ and pairs $x_3$ and $x_5$ to ensure her second point.

It now remains to consider the case where $b_1=y_1$ and $b_2=y_5$. Then Maker plays $x_3$. If Breaker does not immediately answers on the top row, then Maker wins by claiming $x_2$ ensuring one point, and pairing $x_1$ with $x_4$ ensuring another point. By symmetry, assume without loss of generality that Breaker claims $x_1$ or $x_2$. Then Maker claims $x_4$ getting one point, Breaker is forced to answer $x_5$, and Maker wraps up by claiming $z_4$, and pairing $z_3$ with $z_5$.
\end{proof}

The above lemma allows to get a lower bound on the score that is better than the one using a pairing strategy.

\begin{theorem} \label{thm:SCsquares=3}
$\scr(\cG_\squa, 3) \geq 2n/15$.
\end{theorem}

\begin{proof}
It suffices to tile the grid with $G_{3,5}$ subgrids. On each tile, Maker applies her winning strategy that yields at least two points (depending on the parity she will be either first or second player on each tile). The squares between the tiles are considered lost for Maker. Since each $3\times 5$ tile is surrounded by seven squares (east and south), it means that with this strategy, Maker claims at least two points for every subgrid of $15$ squares.
\end{proof}

%%%%%%%%%%%%%%%%%%%%%%%%%%%%%%%%%%%%%%%%%%%%%

%%%%%%%%%%%%%%%%%%%%%%%%%%%%%%%%%%%%%%%%%%%%%
\section{Triangular grid, rhombus (in all three directions), $k=4$}\label{sec:rhombus}

\subsection{Rhombus with $s=1$}

%$\blacktriangleright$ State of the art: $11n/12 \leq \scrp(\cG_\rhom, 1) \leq 23n/24$, $15n/16\leq \scr(\cG_\rhom, 1) \leq 77n/78$

From Theorem~\ref{thm:upper_pairing} we get $\scrp(\cG_\rhom, 1) \leq 23n/24 + o(n)$, from Theorem~\ref{thm:rhombus-maker-s4} we get $\scr(\cG_\rhom, 1) \leq 77n/78 + o(n)$ and from Erd\H{o}s-Selfridge Theorem, we get $\scr(\cG_\rhom, 1)\ge 15n/16 + o(n)$.

\begin{theorem} \label{thm:rhombus-pairing-s1}
$11n/12 \leq \scrp(\cG_\rhom, 1)$.
\end{theorem}

\begin{proof}
Maker pairs the vertices as in Figure~\ref{fig:rhombus-pairing-s1}, marked in red. The only rhombuses that have a chance to have no Maker's vertices are completely contained within one of the green (shaded) areas. 
%Note that within one green area all rhombuses contain the two middle vertices. Hence, unless both of these middle vertices are claimed by Breaker, the green area contains no fully Breaker's rhombuses.
Each green area contains five rhombi -- one central rhombus and the four corner ones. 
\begin{figure}[htb]
    \centering
    \includegraphics[scale=1.6]{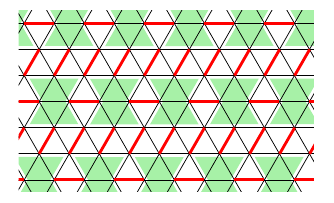}
     \caption{Pairing for Maker, rhombus game with $s=1$.}
    \label{fig:rhombus-pairing-s1}
\end{figure}

Looking at all central rhombi, depicted in gray in Figure~\ref{fig:rhom-central}, all of their corners are paired among themselves in a square-grid-like fashion. Hence, Maker will claim at least half of their total number of vertices, so clearly at most half of them can have no Maker's vertices at all.

As for the corner rhombi, we observe the configuration depicted in Figure~\ref{fig:rhom-corner}. Due to the pairing of the three vertex pairs emphasized in the figure, at most one of these four rhombi can have no Maker's vertices. Now it remains to observe that all the corner rhombi in the board can be partitioned into such configurations of four each. This implies that at most one fourth of the total number of corner rhombi will not be scored by Maker.
\begin{figure}[ht]
    \centering
    \begin{minipage}{0.40\textwidth}
        \centering
        \includegraphics[scale=1.6]{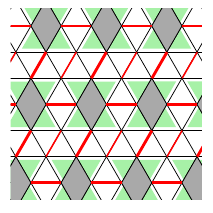}
        \caption{Central rhombi.}
        \label{fig:rhom-central}
    \end{minipage}
    \quad\quad
    \begin{minipage}{0.40\textwidth}
        \centering
        \includegraphics[scale=2]{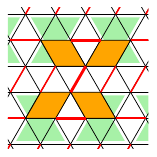}
        \caption{The configuration of four corner rhombi.}
        \label{fig:rhom-corner}
    \end{minipage}
\end{figure}

As the numbers of central and corner rhombi are, respectively, $n/18$ and $2n/9$, the total number of rhombi with no vertices of Maker at the end of the game is at most
$\frac12\cdot\frac{n}{18} + \frac14\cdot \frac{2n}{9} = \frac{n}{12}$.
\end{proof}

\subsection{Rhombus with $s=2$}

%$\blacktriangleright$ State of the art: $19n/36 \leq \scrp(\cG_\rhom, 2) \leq 71n/96$, $\scr(\cG_\rhom, 2) \leq 89n/96$

From Theorem~\ref{thm:upper_pairing} we get $\scrp(\cG_\rhom, 2) \leq 71n/96 + o(n)$ from Theorem~\ref{thm:rhombus-pairing-s3} we get $\scr(\cG_\rhom, 2) \leq 89n/96 + o(n)$.

\begin{theorem} \label{thm:rhom-pairing-LB}
$19n/36 \leq \scrp(\cG_\rhom, 2)$.
\end{theorem}

\begin{proof}
We can see the rhombus grid as the superposition of three square grids. If we apply the pairing depicted in Figure~\ref{fig:squares2} for the square grid and $s=2$ to one of these grids, we get the score of $2/3$ squares for this grid. The same pairing, when transferred to the other two square grids, results in two pairings depicted in Figure~\ref{fig:rhoms2_pairings}. 
\begin{figure}[ht]
    \centering
    $\quad$\begin{minipage}[c]{0.46\linewidth}
        \includegraphics[width = 7cm]{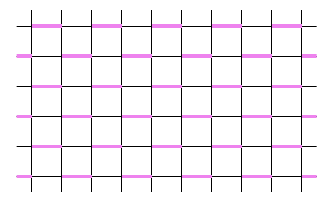}
    \end{minipage}
    $\quad$
    \begin{minipage}[c]{0.46\linewidth}
        \includegraphics[width = 7cm]{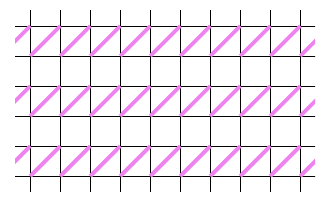}
    \end{minipage}
    \caption{The pairing of for squares $s=2$, when transferred to the two other square grids of the rhombus grid.} $\quad$
    \label{fig:rhoms2_pairings}
\end{figure}

For the second (left) pairing, the squares in each horizontal strip can be matched such that each square is matched to one at distance two from it. Then at least one of the two squares in each matched pair is guaranteed to produce a score. This results in the score of $1/2$ squares on the second grid.

Finally, for the third (right) grid, we partition the squares in two sets -- the squares on the horizontal strips that run on top of the pairs, and the squares that are on the horizontal strips between the pairs. 

The squares on top of the pairs can be matched such that one of the two matched squares is guaranteed to produce a score. For the other squares between the pairs, we can assume the worst-case distribution -- scored squares have 4 Maker's marks, the remaining squares have 1 mark. This way a third of those squares are guaranteed to yield a score, and the total score on the third pairing is $1/2\cdot 1/2 + 1/2\cdot 1/3 = 5/12$.

Hence, the total score in this case is at least $(1/3\cdot 2/3+1/3\cdot 1/2+1/3\cdot 5/12)n = 19n/36$.
\end{proof}

% \begin{figure}[htb]
%     \centering
%     \includegraphics[scale=0.27]{figures/rhoms2_pairing2_onehalf.png}
%      \caption{Temporary picture, to be removed later. \mSt{This is an outcome of the third pairing on the squares, with $1/2$ squares scored. Could be that Breaker cannot perform any better than that. But the best I could prove was $5/12$.}}
%     \label{fig:rhombus-pairing-s3}
% \end{figure}

% \begin{figure}[ht]
%     \centering
%     \begin{minipage}[c]{0.45\linewidth}
%         \includegraphics[width = 8cm]{figures/rhoms2_pairing1_onehalf.png}
%     \end{minipage}
%     \hfill
%     \begin{minipage}[c]{0.45\linewidth}
%         %\includegraphics[width = 8cm]{figures/rhoms2_pairing2_twothird.png} \\ \\
%         %
%         \hspace*{1cm}\includegraphics[width = 6cm]{}
%     \end{minipage}
%     \caption{My best guest for what is the best strategy for breaker against these pairings. \mSt{}There is an outcome of the last pairing (right) with $1/2$ squares scored, see picture.}
%     \label{fig:rhoms2_values}
% \end{figure}

\subsection{Rhombus with $s=3$}

%$\blacktriangleright$ State of the art: $7n/96 \leq \scrp(\cG_\rhom, 3) \leq 5n/16$, $\scr(\cG_\rhom, 3) \leq 17n/36$

From Theorem~\ref{thm:upper_pairing} we get $\scrp(\cG_\rhom, 3) \leq 5n/16 + o(n)$ and from Theorem~\ref{thm:rhom-pairing-LB} we get $\scr(\cG_\rhom, 3) \leq 17n/36 + o(n)$.

\begin{theorem} \label{thm:rhombus-pairing-s3}
$7n/96 \leq \scrp(\cG_\rhom, 3)$
\end{theorem}

\begin{proof}
Maker pairs the vertices as in Figure~\ref{fig:rhombus-pairing-s3}, where all the light blue hexagons have the same pairing on them, with the middle vertex of each hexagon paired with the middle of another hexagon, denoted by red. The pairing on the remaining vertices (not contained in hexagons) is denoted by purple, and it is such that for half of the hexagons there is a purple pair ``wrapped'' around it. 
\begin{figure}[htb]
    \centering
    \includegraphics[scale=1.5]{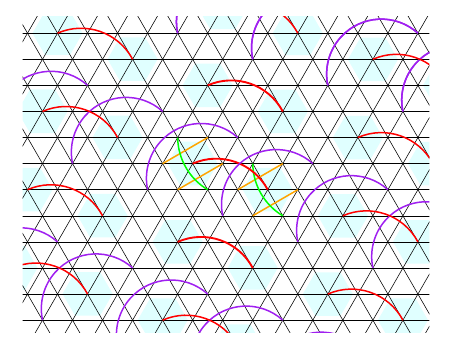}
     \caption{Pairing for Maker, rhombus game with $s=3$. The internal pairing is drawn for only one pair of hexagons.}
    \label{fig:rhombus-pairing-s3}
\end{figure}

First, let note that out of every two paired red vertices Maker will claim one. Let us concentrate on the hexagon around that vertex. Out of the two paired green vertices Maker will claim one, resulting in two neighboring Maker's vertices, one red and one green. Finally, out of two paired orange vertices that can complete a Maker's triangle Maker will claim one, thus claiming a triangle. This triangle is contained in three rhombuses with three Maker's vertices, and these rhombuses do not coincide with such rhombuses obtained from other hexagons. Hence, for every pair of hexagons we obtain three rhombuses with three Maker's vertices.

On top of that, out of every two hexagons with paired centers that have the purple pair ``wrapped'' around them, one will have the center claimed by Maker. Then, looking at the two rhombuses that contain that red center, two orange vertices and one of the paired purple vertices, by simple counting we conclude that one of them has to have at least three Maker's vertices. This further increases the score by one, as such a rhombus does not overlap with the already counted ones.

All in all, the guaranteed score is at least three for every pair of hexagons, and another one for every four hexagons.
\end{proof}

% $\scrp(\cG_\rhom, 3) \geq \frac{n}{24}$. {\color{orange}(Val: if I'm not mistaken. Also, this bound is not very good because the pairings from square also gives that)}
%
% \begin{center}
% \includegraphics[width = 10cm]{figures/rombus4.png}
% \end{center}

\subsection{Rhombus with $s=4$}

%$\blacktriangleright$ State of the art: $n/78 \leq \scr(\cG_\rhom, 4) \leq n/16$.

From Theorem~\ref{thm:upper_pairing} we get $\scrp(\cG_\rhom, 4) \leq n/16 + o(n)$, and from Erd\H{o}s-Selfridge Theorem we get $\scr(\cG_\rhom, 4) \leq n/16 + o(n)$.

%(From Theorem~\ref{thm:rhombus-pairing-s1}, we get $\scr(\cG_\rhom, 4) \leq n/12$, which is worse.)

%\textbf{Conjecture.} $\scrp = 0$ ??
%\todo[inline]{study the conjecture }

\begin{theorem} \label{thm:rhombus-maker-s4}
$n/78 \leq \scr(\cG_\rhom, 4)$.
\end{theorem}

\begin{proof}
Maker uses the tiling with 6-stars depicted in Figure~\ref{fig:rhombus-maker-s4}. If Maker is the first to play within the 6-star, she claims the middle vertex. Then, once Breaker plays within that 6-star, Maker responds by claiming a vertex adjacent to the middle that is the furthest away from the Breaker's move, thus creating a Maker's edge. After the second move of Breaker within the 6-star Maker plays ``on the other side'' of the Maker's edge, thus creating a triangle with two possible ways to extend it to a rhombus. By simply pairing those two extensions Maker claims a whole rhombus.

This guarantees one Maker's rhombus per two 6-stars, completing the proof.
\begin{figure}[htb]
    \centering
    \includegraphics[scale=1.5]{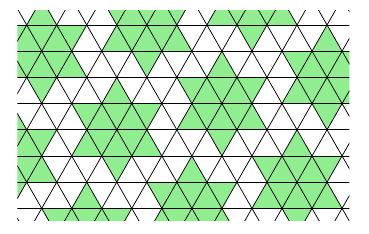}
     \caption{Strategy for Maker, rhombus game with $s=4$.}
    \label{fig:rhombus-maker-s4}
\end{figure}
\end{proof}

%%%%%%%%%%%%%%%%%%%%%%%%%%%%%%%%%%%%%%%%%%%%%
\section{Hexagonal grid, hexagon, $k=6$}\label{sec:hexagon}

\subsection{Hexagon with $s\in\{1,2,5,6\}$}

\begin{observation}
\label{obs:hexa-s1256}
$\scr(\cG_\hexa, 1) = \scrp(\cG_\hexa, 1) = \scr(\cG_\hexa, 2) = \scrp(\cG_\hexa, 2) = n$ and $\scr(\cG_\hexa, 5) = \scrp(\cG_\hexa, 5) = \scr(\cG_\hexa, 6) = \scrp(\cG_\hexa, 6) = 0$.
\end{observation}

\begin{proof}
Using Observation~\ref{obs:1}, it is sufficient to provide a pairing for $s\in\{1,2\}$ such that every hexagon of the grid contains two pairs. One such pairing consists of pairs along all horizontal edges of the grid, see Figure~\ref{fig:hex1}. Note that there is another pairing that also works, see Figure~\ref{fig:hex2}.
\end{proof}

\begin{figure}[ht]
    \centering
    \begin{minipage}{0.40\textwidth}
        \centering
        \includegraphics[scale=1.1]{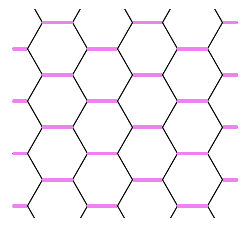}
        \caption{Horizontal pairing.}
        \label{fig:hex1}
    \end{minipage}
    \quad\quad
    \begin{minipage}{0.40\textwidth}
        \centering
        \includegraphics[scale=1.1]{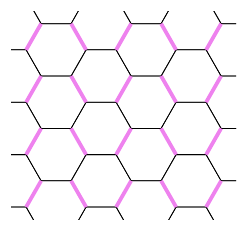}
        \caption{Zig-zag pairing.}
        \label{fig:hex2}
    \end{minipage}
\end{figure}

\subsection{Hexagon with $s=3$}

%$\blacktriangleright$ State of the art: $3n/4 \leq \scrp(\cG_\hexa, 3) \leq \scr(\cG_\hexa, 3) \leq 5n/6$.

From Theorem~\ref{thm:upper_pairing} we get $\scrp(\cG_\hexa, 3) \leq 85n/96$.

Note that if we apply either of the two pairings in figures~\ref{fig:hex1} and~\ref{fig:hex2}, we get $\scrp(\cG_\hexa, 3) \geq n/2$. To see this, note that, for each of the two pairings, Breaker can choose which element of each pair Maker will claim, such that half of the hexagons have only two Maker's elements. But, it turns out that there is a pairing offering a better score.

\begin{figure}[ht]
    \centering
    \includegraphics[scale=1.1]{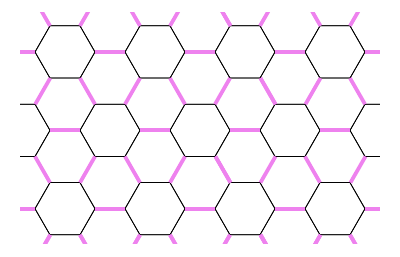}
    \caption{Flower pairing.}
    \label{fig:hex3}
\end{figure}

\begin{theorem} \label{thm:hex-s3-pairing}
$\frac34 n \leq \scrp(\cG_\hexa, 3)$.
\end{theorem}

%\mSt{I did not take the grid boundary into account.}

\begin{proof}
Suppose that Maker applies the pairing depicted in Figure~\ref{fig:hex3}. Then $2n/3$ hexagons have three pairs of paired vertices on their edges, resulting in exactly three Maker's claimed elements in every hexagon. 

We will refer to the remaining $n/3$ hexagons, depicted in Figure~\ref{fig:hex3} with all edges purple, as \emph{flower hexagons}. Note that each vertex of a flower hexagon is paired with a vertex of another flower hexagon. That means that by the end of the game, Maker will have claimed exactly one half of the total number of vertices on the flower hexagons, that is, $\frac12 \cdot \frac{6n}{3} = n$ vertices. 

Let us denote by $m_1$ the number of flower hexagons with at least $s=3$ Maker's vertices, and by $m_2$ the remaining flower hexagons (that have less than $3$ Maker's vertices). Then we have $6m_1+2m_2 \geq n$, with $m_1 + m_2 = n/3$, concluding that $m_1\geq n/12$.
Hence, following the flower pairing Maker can guarantee at least $2n/3 + n/12 = 3n/4$ good hexagons.
\end{proof}

Note that the lower bound in Theorem~\ref{thm:hex-s3-pairing} is the best we can achieve with the flower pairing, as Maker can be given one vertex in each pair such that the number of good hexagons is exactly $3n/4$.

As a direct consequence of Theorem~\ref{thm:hex-s4-Maker-nopairing}, we get the following upper bound.

\begin{corollary} \label{cor:hex-s3-ub}
$\scr(\cG_\hexa, 3) \leq \frac{5}{6} n$.
\end{corollary}

\subsection{Hexagon with $s=4$}

%$\blacktriangleright$ State of the art: $\frac{1}{10} n\leq \scrp(\cG_\hexa, 4)$, $\frac16 n\leq \scr(\cG_\hexa, 4) \leq \frac{1}{4} n$

From Theorem~\ref{thm:hex-s3-pairing}, we get $\scr(\cG_\hexa, 4) \leq n/4 + o(n)$.

\begin{theorem} \label{thm:hex-s4-Maker-pairing}
$\frac{n}{10}\leq \scrp(\cG_\hexa, 4)$.
\end{theorem}

\begin{proof}
We give a pairing strategy for Maker. To do that, we use the tiling with hexagon triplets shown in Figure~\ref{fig:hex-pairing-tiling}. In each row, triplets are coupled two-by-two, from left to right. Within each couple of triplets, the pairing is depicted in Figure~\ref{fig:hex-zoom-in}, where lines in green, red and purple denote the paired vertices.
\begin{figure}[htb]
        \centering
        \includegraphics[scale=1]{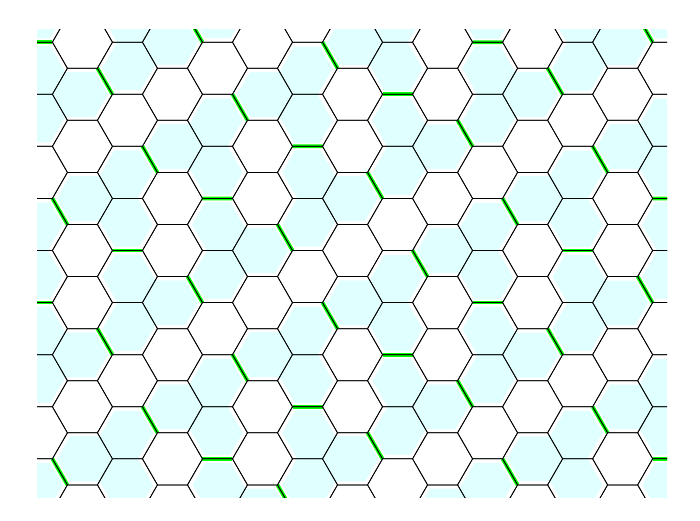}
        \caption{Tiling with hexagon triplets, for the proof of Theorem~\ref{thm:hex-s4-Maker-pairing}.}
        \label{fig:hex-pairing-tiling}
\end{figure}

\begin{figure}[htb]
        \centering
        \includegraphics[scale=1.3]{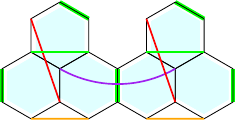}
        \caption{Pairing within two coupled triple hexagons.}
        \label{fig:hex-zoom-in}
\end{figure}

It remains to prove that Maker can score at least one hexagon in every couple of hexagon triplets, and that will result in a score of one per every ten hexagons. To show this, first note that the purple pair ensures that Maker will claim the vertex at the center of one of the two triplets. 

From now on we concentrate on the triplet of hexagons with Maker's center. One of the two vertices in the red pair is going to be claimed by Maker. If it is the upper vertex, then the upper hexagon will be scored by Maker. If, however, it is the lower vertex, then the low orange pair ensures that she will score on one of the two lower hexagons.
\end{proof}

\begin{theorem} \label{thm:hex-s4-Maker-nopairing}
$\frac{n}{6}\leq \scr(\cG_\hexa, 4)$.
\end{theorem}
\begin{proof}
We tile the hexagonal grid with $2\times 2$ hexagonal subgrids as depicted in Figure~\ref{fig:hexa-s4-tiling}. Maker combines a pairing strategy on the red pairs, shown in the figure to be shared between the neighboring subgrids, and a more sophisticated strategy on the remaining vertices of each subgrid.
\begin{figure}[htb]
    \centering
    \includegraphics[scale=1]{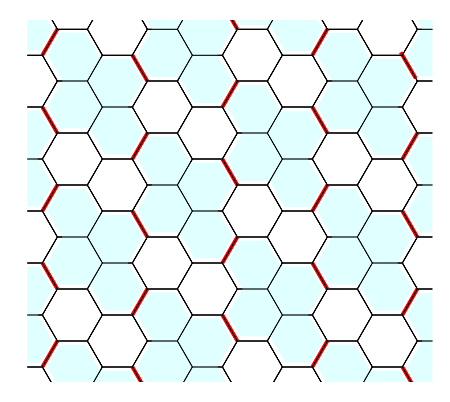}
     \caption{Tiling for Maker, hexagon game with $s=4$.}
    \label{fig:hexa-s4-tiling}
\end{figure}

On each $2\times 2$ hexagonal subgrid Maker will respond as second, meaning that she will repeatedly wait for the Breaker's move in that subgrid to respond. In her first move, Maker will claim one of the two central vertices marked yellow, as depicted in Figure~\ref{fig:hexa-s4-subgrid} on the left.
\begin{figure}[ht]
    \centering
    \begin{minipage}[c]{0.28\linewidth}
        \includegraphics[page=1,width = 3cm]{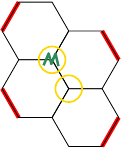}
    \end{minipage}
    \begin{minipage}[c]{0.28\linewidth}
        \includegraphics[page=2,width = 3cm]{figures/hexa-M-123-new.pdf}
    \end{minipage}
    \begin{minipage}[c]{0.28\linewidth}
        \includegraphics[page=3,width = 3cm]{figures/hexa-M-123-new.pdf}
    \end{minipage}
    \caption{Maker's strategy on a $2\times 2$ hexagonal subgrid, in the hexagon game with $s=4$.}
    \label{fig:hexa-s4-subgrid}
\end{figure}
W.l.o.g.~we can assume that it is the one marked with `M'. Then in her second move, she claims one of the three vertices marked yellow in the middle figure, say the one marked with `M'. This way, the first two Maker's moves are common for two neighboring hexagons. Finally, in her third move Maker will claim one out of the four remaining vertices on those two hexagons, marked yellow in the right figure, thus scoring on one of the hexagons (having in mind that Maker will claim a vertex on every red pair).

Note that in each of Maker's moves, the number of candidate vertices marked yellow is strictly larger than the number of vertices Breaker claimed within that subgrid up to that point.
\end{proof}

\section*{Conclusion}

This paper advances the understanding of scoring positional games, building on the initial work of~\cite{scoringPG}. We studied the games on grids, and utilized several techniques that yield both upper and lower bounds, for optimal as well as pairing strategies.
Since most of these bounds are not tight, a number of open problems remain. Some situations in which the gap between the lower and upper bounds is relatively small are particularly notable, such as, e.g., for $\scr(\cG_\hexa,3)$.

One problem we find especially intriguing is the full rhombus game for pairing Maker, where we conjecture that $\scrp(\cG_\rhom,4)=0$. Given that we established a strategy guaranteeing Maker a score of at least $n/78+o(n)$, it would be very interesting to determine whether there exists a pairing strategy that gives Maker any positive score. This appears to be a difficult question, since no general approach is known for problems of this type, other than exhaustive exploration.

\section*{Acknowledgement}
This research was supported by the ANR project P-GASE (ANR-21-CE48-0001-01), and it was supported by COST Action CA22145 GameTable. Milo\v{s} Stojakovi\'c was partly supported by the Science Fund of the Republic of Serbia, Grant \#7462: Graphs in Space and Time: Graph Embeddings for Machine Learning in Complex Dynamical Systems (TIGRA), and partly supported by the Ministry of Science, Technological Development and Innovation of the Republic of Serbia (grants 451-03-33/2026-03/200125 \& 451-03-34/2026-03/200125).

\end{document}